\documentclass[]{article}

\usepackage{amsmath,amssymb,amsfonts}
\usepackage{graphicx}
\usepackage{color}
\usepackage{appendix}
\usepackage{authblk}
\usepackage{amsthm}
\usepackage{tikz}

\usepackage[english]{babel}
\usepackage[utf8]{inputenc}

\usepackage{cite}

\usepackage{geometry}
\newtheorem{theorem}{Theorem}

\newtheorem{definition}{Definition}

\newcommand{\nn}{\nonumber}

\def\a{\alpha}
\def\b{\beta}

\def\d{\delta}

\def\th{\theta}

\def\r{\rho}

\def\s{\sigma}

\def\p{\psi}

\def\La{\Lambda}
\def\P{\Psi}

\def\<{\langle}
\def\>{\rangle}
\def\Tr{Tr}

\def\Tr{{\rm Tr}}
\def\tr{{\rm tr}}

\def\cT{{\mathcal{T}}}

\def\vn{{\vec{n}}}
\def\vm{{\vec{m}}}
\def\hS{{\mathcal{S}}}
\def\hs{{\tilde{S}}}
\def\perm{{\textrm{perm}}}

\begin{document}
	
	\title{Partial Distinguishability as a Coherence Resource in Boson Sampling}
	
	\author[1]{Seungbeom Chin \thanks{sbthesy@gmail.com}}
	\author[2,3]{Joonsuk Huh \thanks{joonsukhuh@gmail.com}}
	\affil[1]{Department of Electrical and Computer Engineering, Sungkyunkwan University, Suwon 16419, Korea}
	
	\affil[2]{Department of Chemistry, Sungkyunkwan University, Suwon 16419, Korea }
	
	\affil[3]{ SKKU Advanced Institute of Nanotechnology (SAINT), Sungkyunkwan University, Suwon 16419, Korea}
	\maketitle

	\begin{abstract}
	Quantum coherence is a useful resource that is consumed to 
	accomplish several tasks that classical devices are hard to fulfill. Especially, it is considered to be the origin of quantum speedup for many computational algorithms. In this work, we interpret the computational time cost of boson sampling with partially distinguishable photons from the perspective of coherence resource theory. With incoherent operations that preserve the diagonal elements of quantum states up to permutation, which we name \emph{permuted genuinely incoherent operation} (pGIO), we present some evidence that the decrease of coherence corresponds to a computationally less complex system of partially distinguishable boson sampling. Our result shows that coherence is one of crucial resources for the computational time cost of boson sampling. We expect our work presents an insight to understand the quantum complexity of the linear optical network system.
\end{abstract}

{\bf Keywords:} coherence resource theory, linear optics, boson sampling, partial distinguishability, permuted genuinely incoherent operation

\section{Introduction}
Boson sampling (BS) \cite{aaronson2011computational} is a non-universal quantum computing device that can be easily realized with the quantum linear optical network (LON). The BS process has some practical advantage over known quantum algorithms using universal quantum computers, e.g., Shor algorithm. Indeed, it can be implemented with a more feasible system (LON) than the algorithms based on universal quantum computers. In the original BS setup proposed in Ref. \cite{aaronson2011computational}, $N$ single-photon states are initially prepared in $M$ ($\gg N $) input modes. The photons are injected into a LON that generates interference by unitary operations and then are detected in $M$ output modes. As a result, the transition amplitude between the product of single photon states is hard to simulate with classical Turing machines. Therefore, BS seems to be a strong candidate to refute the Extended Church-Turing Thesis (ECT), which conjectures that a Turing machine can efficiently simulate any efficient computational process in the real world.

However, the BS process should fulfill several conditions for its computational hardness: low photon density ($M \gg N$), complete photon indistinguishability, and randomness of unitary operations. Most of all, the complete photon indistinguishability condition is hard to meet since photons usually carry some internal degrees of freedom that makes them \emph{partially distinguishable} in experimental realizations \cite{toppel2012all}. There have been many quantitative approaches to analyze the multiphoton interference phenomena of partial distinguishable photons in LON \cite{tan20133, tillmann2015generalized, de2014coincidence, rohde2015boson,tamma2014multiboson, tamma2015v, shchesnovich2014,shchesnovich2015partial,tichy2015sampling}, in which the \emph{distinguishability matrix} is introduced to evaluate the mutual distinguishability of each particle to others. The transition probability can be calculated in terms of the partial distinguishable matrix, which is denoted by $\hS$ in Ref. \cite{tichy2015sampling}. The authors of Ref. \cite{renema2017efficient} sought a range of partial distinguishability under which a classical simulation of BS becomes efficient. They first considered a special form of $\hS$ that has the form of the interpolation of the fully distinguishable and fully indistinguishable cases with one continuous real parameter $x$ ($0 \le x \le 1$). Then they applied the result to the generalized $\hS$ by imposing the efficient upper bound that also can be controlled by the parameter. Actually, when the distinguishability is evaluted with just one parameter, we can assert that the amount of $x$ directly determines the degree of indistinguishability (DOI).

However, when $\hS$ is required to be in a generalized form with more parameters, it is not that straightforward to determine DOI for the given matrix. For such a case, we need to find some scalar measures that can compare DOI between different forms of $\hS$.
Therefore, for a comprehensive discussion on the relation between particle distinguishability and BS complexity, embracing that in Ref. \cite{renema2017efficient}, we first need to present a generic and rigorous criterion of DOI for the generalized $\hS$.
In this work, we approach this problem by exploiting the \emph{quantum coherence resource theory}, which was quantitatively formalized first in Ref. \cite{baumgratz2014quantifying} (for a general review for the theory, see Ref. \cite{streltsov2017colloquium}).

Quantum resource theories provide the quantitative criteria to compare the amount of resources in different quantum systems.  They also equip one to analyze relations between the variation of resources and the efficiency (speed or complexity) of some quantum tasks (see, e.g., Ref.~\cite{chitambar2018quantum}). Among them,
the coherence resource theory has incoherent states as free states (which possess no quantum feature) and incoherence operations as free operations (which decrease the quantum feature of given states when applied to the states). Our idea is to treat the amount of particle distinguishability for $\hS$ as quantum coherence. On the other hand, since $\hS$ should remain in the form of a Gram matrix with all diagonal elements  1, we need a very restricted set of incoherent operations to apply the coherence resource theory to our current system. We name such a class \emph{the permuted genuinely incoherent operation} (pGIO). This is a slightly extended set of operations from the genuinely coherent operation (GIO) \cite{de2016genuine}, by adding permutations. An intriguing property of pGIO is that it is the intersection of the strictly incoherent operation (SIO) set \cite{yadin2016quantum} and the fully incoherent operation (FIO) set \cite{de2016genuine}.

The main focus of this work is to identify the behavior of transition probability of partially indistinguishable photons under a pGIO. We expect that \emph{when a pGIO is applied to a distinguishability matrix $\hS$, this operation can be exploited to decrease the computational time cost of BS with partially distinguishable photons}. Examples to support this conjecture are presented here.  By introducing some permuted genuine (pG) coherence monotones, we will analyze the behavior of the upper bound of the transition probability and the runtime for exactly simulating the transition probability of a given BS system. By interpreting the particle indistinguishability in the framework of coherence resource theory, one can state that \emph{quantum coherence plays a role of resource for the passive scattering process of linear optics to become classically hard simulate}. We expect our present work would provide an insight into understanding the computational complexity of the linear optical network system with quantum resource theories.

Our work is organized as follows: In Section \ref{partial}, we  briefly review the concept of partial distinguishability matrix and its relation to the scattering process in LON. In Section \ref{pGIO}, we define pGIO and discuss some interesting properties of the operation class. In Section \ref{pGIOapp}, we analyze the relation between pGIO and the computational time cost of partially distinguishable BS. In Section \ref{con}, we present a summary and suggest some possible future research.

\section{Partial distinguishability in LON}\label{partial}

We first introduce the concept of $N\times N $ distinguishability matrix $\hS$ and explain how it affects the transition probability in multi-mode linear optical network systems \cite{tichy2015sampling}. 
We discuss the role of partial distinguishability among photons in the original Fock state BS by Aaronson and Arkhipov \cite{aaronson2011computational}, in which each input and output mode contains no more than one photon (the complexity of the arbitrary photon distribution case with fully indistinguishable photons is discussed in Ref. \cite{yung2016universal, chin2017majorization,chin2018gen}).
The relation between distinguishability matrix $\hS$ and density matrix coherence is also explained here.

\subsection{Partial distinguishability matrix $\hS$ and the transition probability of LON }

All the possible internal degrees of freedom (e.g., angular frequency, polarization) of the $i$th photon ($1\le i \le N$) can be described in general with a normalized ``internal'' state $|\phi_i\>$. 
Then the mutual distinguishability of $N$ photons is represented with the \emph{distinguishability matrix} $\hS$ with the elements
\begin{align}
\hS_{ij} = \< \phi_i |\phi_j \>.
\end{align}
We can directly see that $\hS$ is a Gram matrix, and hence positive semidefinite (PSD). 
Since $|\phi_i\>$ are non-orthogonal normalized states, we have $0\le |\hS_{ij}|\le 1$ and $\hS_{ii}=1$ for all $i$. When all internal states are orthogonal to each other (all particles are completely distinguishable), we have $\hS = \mathbb{I}$ ($\hS_{ij}=\d_{ij}$). On the other hand, when all internal states are proportional to each other (completely indistinguishable), we have $\hS_{ij} = 1$ for all $i$ and $j$.

With a nontrivial $\hS$,
the transition probability of Fock state BS is not expressed as an absolute square of transition amplitude: the probability for a post-selected photon distribution is in general given by \cite{tichy2015sampling}
\begin{align}\label{transitionprobability}
P(\vn,\vm) = \sum_{\s\in S_N }\Big(\prod_i \hS_{i, \s_i} \Big)\perm(V\odot V^*_{\s, \mathbb{I}}).
\end{align}
Here $V$ is the submatrix of the linear optical unitary operation $U$ that actually generates mode interference. More specifically, when the input and output photon distribution vectors are given by $\vn = (n_1,n_2,\cdots , n_M)$ and $\vm = (m_1,m_2,\cdots , m_M)$ respectively, with  $N$ ($=\sum_{i}n_i = \sum_i m_i$) photons and $M$ modes,    we have $V=U_{\vn,\vm}$ ($N\times N$ submatrix of $U$ that has $n_i$ ($m_i$) of the $i$th row (column) of $U$). $V^*_{\s, \mathbb{I}}$ is  the complex conjugate of $V$ with columns permuted along a specific permutation $\s$. The entrywise Schur product (or Hadamard product) is denoted by $\odot$, and the summation of permutations is over all elements of the permutation group $S_N$. When particles are fully indistinguishable ($\hS_{ij} =1$ for all $i$ and $j$), $P(\vn,\vm)= |\perm(V)|^2$. When particles are fully distinguishable ($\hS_{ij}=\d_{ij}$), $P(\vn,\vm) = \perm(|V|^2)$.

\subsection{$\hS$ and coherence}

We can understand the relation between $\hS$ and density matrices from the viewpoint of coherence, which is indispensable  for the connection of our system to the coherence resource theory. 
Coherence is a quantum feature that indicates the degree of superposition among orthogonal quantum states. There exists a duality between the quantum coherence and path indistinguishability of multi-slit (or multi-mode) interference phenomena \cite{bera2015duality,bagan2016relations, chin2017generalized}.
To see the relation in more detail, consider a quanton (quantum wave/particle dual existence) interferes by a $N$ dimensional multi-slit. A quanton state that passes through the $i$th slit is denoted as $|i\>$ ($1\le i \le N$). Then the most general quanton state is given by
\begin{align}
|\p_{in}\>= \sum_{i=1}^N c_i|i\>. \qquad (\sum_{i}|c_i|^2 = 1)
\end{align}
$\{|i\>\}_{i=1}^{N}$ constructs an orthonormal basis set. To know which slit a quanton passes, we need a detecter that is entangled to the quanton, which makes the total state including the quanton and detector as 
\begin{align}
|\P\> = \sum_{i=1}^Nc_i|i\>\otimes |\phi_i\>_{D} 
\end{align}
$|\phi_i\>_D$ is a state for a detector attatched to the $i$th slit, which is normalized but not necessarily orthogonal according to its resolution. To acquire the path information of a quanton, we need to partial trace over the detector state,
\begin{align}
\r_r = \tr_D(|\P\>\<\P|) = \sum_{i,j}c_ic_j^* \<\phi_j|\phi_i\>|i\>\<j|.
\end{align}  Then one can show that the coherence of $\r_r$ determines the path distinguishability of the quanton \cite{bera2015duality, bagan2016relations, chin2017generalized}. Note that when $|\p_{in}\>$ is maximally coherent, i.e., $|\p_{in}\> =\sum_{i=1}^N\frac{1}{\sqrt{N}}|i\>$, $\r_r$ is proportional to the complex conjugate of $\hS$:
\begin{align}
\r_{r} \xrightarrow{\textrm{maximally coherent}} (\r_r)^{max} =\frac{1}{N} \hS^* . 
\end{align}
Considering the complex conjugation comes from the definition of $\hS$, we can state that the partial distinguishability of bosons for given optical modes can be equivalently treated as the spatial coherence of the bosons among the modes. 
This relation renders the application of coherence resource theory to the multi-mode scattering of partial distinguishable particles. 
We define $\hs \equiv (\r_r)^{max}=\frac{1}{N}\hS^{*}$, which satisfies the conditions for a density matrix.

\section{Coherence resource theory and permuted genuinely  incoherence operation (pGIO)}\label{pGIO} 
Coherence resource theory \cite{abergj, baumgratz2014quantifying, winter2016operational, yadin2016quantum} recently has drawn extensive
attention, and it turns out that the coherence enhances the efficiency of various quantum computational tasks such as Deutsch-Jozsa algorithm \cite{hillery2016} and Grover algorithm \cite{shi2017,chin2017coherence}. 
Coherence depends on a specific set of computational bases, and we define an \emph{incoherent state} in $d$-dimensional Hibert state $\mathcal{H}$ as a diagonalized state in the computational basis set $\{|i\>\}_{i=1}^d$. 
The standard incoherent operations (IO) \cite{baumgratz2014quantifying} corresponds to the following Kraus decomposition:
\begin{align}
\La_I[\r] = \sum_n K_n \r K_n^{\dagger}
\end{align}
with $\sum_n K_n K_n^{\dagger}=\mathbb{I} $ and $ K_n\hat{\delta} K_n^{\dagger}/tr[K_n\hat{\d}K_n^{\dagger}] =\hat{\delta '}$ ($\delta,\delta'$ are both incoherent states). The Kaus operators are explicitly expressed as $K_n = \sum_i c_n^i|f_i^n\>\<i|$ ($f^n$ is a function that sends $i$ to $i'$, not necessarily one-to-one). On the other hand, other kinds of incoherent operations have been suggested according to various physical motivations. Strictly incoherent operations (SIO) are those which cannot use the coherence in input states, which has $K_n = \sum_i c_n^i|\s_i^n\>\<i|$ ($\s^n$ is  a permutation, hence one-to-one now) \cite{winter2016operational, yadin2016quantum}. Genuinely incoherent operations (GIO) preserve all incoherent states, which has $K_n = \sum_i c_n^i|i\>\<i|$ \cite{de2016genuine}. Fully incoherent operations (FIO) have the most general form that are incoherent for all $K_n$, with  $K_n = \sum_i c_n^i|f_i\>\<i|$, i.e., $K_n$ have the same matrix form for all $n$ \cite{de2016genuine} (a GIO is naturally an FIO) \footnote{Other crucial classes of incoherent operations are  physical incohent operations (PIO) \cite{chitambar2016comparison,chitambar2016critical}, dephasing-covariant incoherent operations (DIO) \cite{chitambar2016comparison,chitambar2016critical}, and translationally-invariant operations (TIO) \cite{gour2009measuring,marvian2016quantify, marvian2014extending, marvian2016quantum}.}.  

\subsection{Permuted genuinely incoherent operations (pGIO)}

The close relation between coherence and indistinguishability was first pointed out in Ref. \cite{mandel1991} for the case of one photon in two modes, but the application of coherence to the partial distinguishability case requires a very different mathematical approach.
Here we should analyze the behavior of the transition probablity Eq. \eqref{transitionprobability} when the coherence of $\tilde{S}$ changes according to some incoherence operations.    
However, since the diagonal elements of $\hs$ should be preserved under any operation as $1/N$ for all $i$, we need a special kind of incoherent operations that satisfy this restriction to analyze our physical system. Here we suggest \emph{permuted genuinely incoherent operations} (pGIO) as such a class of incoherence operations:
\begin{definition}
	Permuted genuinely incoherent operations (pGIO) are those which preserve the diagonal elements of given states within permutation. 
\end{definition}
It is direct to note that the set of pGIO includes that of GIO. The following inclusion relation also hold:
\begin{theorem}
	The set of $pGIO$ is the intersection of $SIO$ and $FIO$ (Figure 1).
	\begin{proof} The Kraus operators that satisfy both the conditions for SIO and  FIO are expressed as $K_n= \sum_i c_n^i|\s_i\> \<i|$, which can be decomposed as 
		\begin{align}
		K_n= \sum_i c_n^i|\s_i\> \<i| =\sum_{i}|\s_i\>\<i| \sum_jc_n^j|j\>\<j|.
		\end{align}
		The final expression represents the definition of pGIO.
	\end{proof} 
\end{theorem}
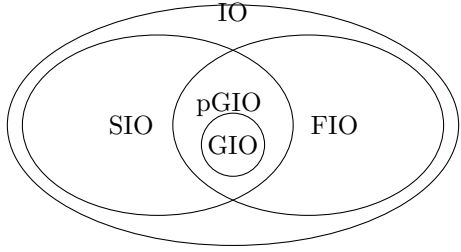
\begin{figure}
	\centering
	\begin{tikzpicture}
	\draw (0,0) circle [x radius=3cm, y radius=1.6cm]  node[align=center, above] {IO\\$ $ \\$ $\\$ $ };
	\draw (-1,0) circle [x radius= 18mm, y radius=12mm] node[align= center] {SIO$\qquad$};
	\draw (1,0) circle [x radius=18mm, y radius=12mm] node[align= center] {$\qquad$FIO} node[align=left, above] {pGIO$\qquad \qquad\quad\quad  $};	
	\draw (0,-0.26) circle [radius=4.2mm] node[align=center] {GIO}; 
	\end{tikzpicture}
	\caption{A Venn diagram for the relation of incoherent operations.}
	\label{Venn}
\end{figure}


\begin{theorem}
	Any state that has the form of $\hs$ can be obtained by taking pGIOs to a maximally coherent state $\r_{M}$ ($\r_{M ij}=e^{i(\th_i-\th_j)}/N$ for all $i$ and $j$).  
\end{theorem}
\begin{proof} From Theorem 2 of Ref. \cite{de2016genuine},
	a density matrix $\r = \sum_{ij}\r_{ij}|i\>\<j|$ transforms under pGIO to $\r' =\sum_{ij}(A\odot\r)_{\s_i\s_j}|i\>\<j|$ where $A$ is a Hermitian positive semidefinite matrix (therefore a Gram matrix) with $A_{ii}=1$ for all $i$. It is straightforward that a matrix whose elements are the Hadamard product of two Gram matrices is also a Gram matrix \cite{marvian2013theory}. A Gram matrix hence transforms to another Gram matrix under a pGIO. As a result, any state of the form $\hS$ can be obtained with pGIOs from $\r_M$. 
\end{proof}
Note that, by permutation, an unspeakable resource theory (GIO) becomes a speakable theory (pGIO) \footnote{The speakable resources are independent of the physical encoding, i.e., all basis are equivalent, while the unspeakable depend on the specific degrees of freedom. For a more detailed explanation, see Ref. \cite{marvian2016quantify, chitambar2018quantum}. }.

To see the relation between pGIO and  partially distinguishable BS, it is convenient to find quantum quantities that can evaluate the degree of coherence (permuted genuine coherence motonones). 
Following the definitions for established incoherence operations including IO \cite{baumgratz2014quantifying}, the conditions that a permuted genuine (pG) coherence monotone $C_{pG}$ must fulfill are as follows: 

(pG1) Nonnegativity: $C_{pG}(\r) \ge 0$, and $C_{pG}(\r)=0$ if and only if $\r$ is incoherent.
(pG2) Monotonicity:  $C_{pG}(\La(\r)) \le C_{pG}(\r)$ for any pGIO $\La$.
(pG3) Strong monotonicity: $C_{pG}$ does not increase under selective operations for any Kaus operator set $\{K_n \}$, i.e,  $\sum_nq_nC_{pG}(\r'_n) \le C_{pG}(\r)$ with  $q_n= \Tr(K_n\r K_n^{\dagger})$ and  $\r'_n = K_n\r K_n^{\dagger}/q_n$.
(pG4) Convexity: $C_{pG}(\sum_np_n \r_n) \le \sum_{n}p_nC_{pG}(\r_n)$.

For a quantity to be a pG coherence monotone, it must mininally satisfies (pG1) and (pG2).

To understand the role of pGIO in BS problems, we need to find monotones that also have  straightforward relations to the scattering process of BS. 
Here we suggest three such pG coherence monotones, $\mathcal{N}(\r)$ (the number of nonzero entries of $\r$), $\perm(|\r|)$ (permanent of the matrix whose elements are the absolute values of the entries of $\r$), and $J_{a}(\r)$ which is defined as follows:
\begin{definition}\label{def2}
	With a scalar value $J_{\s^a}(\r)\equiv |\prod_{i=1}^N \r_{i\s^a_i}|$ where $\s^a$ is a permutation that have (N-a)-fixed points (in other words,  $\s^a$ is a permutation without changing $(N-a)$ elements among $N$ elements), we define $J_{a}(\r) =\max(J_{\s^a}(\r))$, i.e., the maximal value of $J_{\s^a}(\r)$ among all possible permutations with $(N-a)$ fixed points. 
\end{definition}
Then we have the following theorem.
\begin{theorem}\label{monotones}
	$\mathcal{N}(\r)$, $\perm(|\r|)$, and $J_a(\r)$  are  pG coherence monotones that satisfy (pG1) and (pG2). 
\end{theorem}
\begin{proof}
	1. $\mathcal{N}(\r)$:	$(pG1)$ is trivially satisfied. Since all zero entries of the density matrix $\r$ only change their places but not the amount, $(pG2)$ is also true.\\
	2. $\perm(|\r|)$: (pG1) is again trivially satisfied. 
	Under a pGIO, elements $\r_{ij}$ of a density matrix $\r$ are transformed to $\r'_{ij} = (A\odot \r)_{\s_i\s_j}$, with $|A_{ij}| \le 1$ for all $i$ and $j$. Therefore, the inequality $\perm(|\r'|)\le \perm(|\r|)$ always holds, and $(pG2)$ is satisfied.\\
	3.	 $J_{a}(\r)$: it is straightforward to see that $J_{\s_a}(\r)$ satisfies $(pG1)$ and $(pG2)$ using $|A_{ij}|\le 1$. Then even if the maximal permutation changes for $\max(J_{\s_a}(\r) )$, $J_{a}(\r')$ cannot be greater than $J_{a}(\r)$.
\end{proof}

One can ask about the actual physical implication of pGIO in the multimode scattering process of partially distinguishable photons. It can be answered by considering that pGIO on  $\hs$ is equivalent to the Hadamard product of two Gram matrices. Therefore, for a given internal state $|\phi_i\>$ that determines $\hS$, we can state that \emph{a pGIO on $\hS$ is to attach additional degrees of freedom, e.g., $|\p_i\>$ so that a new internal state becomes $|\phi_i\> \otimes |\p_i\>$}. And the particles are likely to be more distinguishable with more degrees of freedom, and the Gram matrix for the new initial state is the updated partial distinguishability matrix $\hS$.

Now we are ready to investigate the relation between pGIO and the computational cost of BS with partially distinguishable photons.

\section{pGIO in partially distinguishable boson sampling}\label{pGIOapp}

In this section, we analyze the behavior of the partially distinguishable BS transition probability under pGIO.
When a pGIO $\La$ is applied to a given partially distinguishable matrix $\hS$, the new bosonic system becomes more distinguishable (see the end of Sec. \ref{pGIO}) and pG monotones decrease by definition (see pG1 to pG3) as well. Hence, if the computational time cost of BS with partially distinguishable photons decreases as the pG monotones decreases, the relation supports our surmise that the decrease of coherence in the system (which is equivalent to the increase of distinguishability) depletes the computational complexity of BS. Thus one can state that the distinguishability of photons is exploited to reduce the computational time cost of BS.

The time cost of classically simulating the BS scattering process can be understood from two aspects, i.e., approximate and exact simulations. We can summarize the results of this section in intuitive expressions as:
\\
\\
1. The classical approximation of the BS transition amplitude becomes more efficient as pGIOs are applied to $\hS$, i.e., as photons become more distinguishable (\ref{approx}).
\\
2. The runtime for exactly simulating the BS transition amplitude becomes shorter as pGIOs are applied to $\hS$ (\ref{exact}).
\\ 

The monotones we have introduced in Section \ref{pGIO} ($\mathcal{N}$, $perm(|\r|)$ and $J_a$) are used to examine the relation between transition probability and pGIO.

\subsection{Approximating the transition probability and $J_a$}\label{approx}
Various quantities have been suggested as the DOI for multi-boson scattering experiments \cite{de2014coincidence, shchesnovich2014, shchesnovich2015partial, tichy2015sampling, walschaers2016many, brunner2017signatures} from different physical perspectives. Most of all, it is shown in Ref. \cite{tichy2015sampling} that  $\perm(|\hS|)$ is directly related to the upper bound of $P(\vn,\vm)$ (See Appendix \ref{bound} for a detailed analysis on the bound). 
On the other hand, one can consider a tighter bound of $P(\vn,\vm)$ that is more easily saturated by phase control. The bound divides the effect of indistinguishability from that of distinguishabilty and also reveals the monotonic effect of pGIO on $\hS$ manifestly:
\begin{align}\label{tighter}
P(\vn,\vm) \le & \sum_{\s\in S_N} \big|\perm(V\odot V_{\s,\mathbb{I}}^*)\big| J_\s  \quad \Big( \le P_\mathbb{I}[\perm(|\hS|)] \Big)  \nn \\
=& \perm(V\odot V^*)+ \sum_{\s\in S_{N}, \s\neq \mathbb{I}} \big|\perm(V\odot V_{\s,\mathbb{I}}^*)\big| J_\s, 
\end{align}
where $J_\s \equiv |\prod_i \hS_{i,\s_i}|$. Note that the first term in the last line of Eq. \eqref{tighter} corresponds to the classical contribution (distinguishable scattering), and the second term to the nonclassical contribution (path interference by indistinguishability).

Since $J_\s$ is multiplied by each term that represents the effect of interference in the last term of Eq. \eqref{tighter}, with Theorem \ref{monotones}, we can see that \emph{the impact of interference in LON decreases under any pGIO}. We speculate that the reduction of interference results in a computationally less complex scattering process  (see, e.g., \cite{lloyd1999quantum, stahlke2014quantum}) \footnote{This can be compared to the coherence theory of wave-particle duality in multi-slit experiments \cite{bera2015duality, bagan2016relations, chin2017generalized, biswas2017interferometric}. These researches showed that the interference phenomena (wave-like property) of a quanton through a multi-slit path inceases as the degree of coherence increase, which is analogous to our case with multimode linear optical network. }. The following analysis supports this assumption.


Using $S_{ii} =1$ for all $i$, $J_\s$ can be ordered along the number of fixed points in permutations as $J_{\s^{}a}$ with $a=0,2,\cdots, N$ (see Definition \ref{def2}). For example, when only two points $i$ and $j$ are permuted ((N-2)-points are fixed), $J_{\s^2} =|\hS_{ij}|^2$, etc. Therefore, we can enumerate Eq. \eqref{tighter} as
\begin{align}\label{pordered}
P(\vn,\vm) \le  \sum_{a}\sum_{\s^a}J_{\s^a} \perm(V\odot V_{\s^a,\mathbb{I}}^*) \equiv \sum_a Z_a,
\end{align} 
Since the order of $|\hS_{ij}|$ ($\le 1$) increases as $a$ increases,  $Z_a$ with lower $a$ makes a greater contribution to the probablity on average \cite{tichy2015sampling,rohde2015boson, renema2017efficient}.

The condition for efficiently approximating the transition probability with the lowests $k$ term of $Z_a$, i.e., $P_k = \sum_{a}^{k\ge N}Z_a$, is given in Ref. \cite{renema2017efficient}.  Since the scattering matrix is chosen totally randomly, the inequality in Eq. \eqref{pordered} becomes an equality for real $x_{ij}$ without loss of generality. The $k$-photon approximation for Eq. \eqref{pordered} is obtained by setting max$(J_{\s^{k}})^{1/k}\equiv (J_k)^{1/k}\equiv x$ ($x $ is real and $0 \le x \le 1$). 
As $x$ becomes small, which is achieved by a pGIO,  the approximation becomes efficient with lower k. Indeed, FIG. 4 of Ref. \cite{renema2017efficient} shows that the approximation level $k$ decreases as $x$ decreases.
Since $x= J_{k}^{1/k}$ is pG coherence monotone by Theorem \ref{monotones}, we can state that \emph{the pGIO on an arbitary $\hS$ decreases the computational time cost of approximating the transition process.} 


\subsection{Exact classical algorithm for simulating transition probability  and $\mathcal{N}(\hs)$}\label{exact}
Here we show that the decrease of $\mathcal{N}(\hs)$ permits a less expensive algorithm to exactly simulate the transition probability.
The transition probability of the partial distinguishable BS (Eq. \eqref{transitionprobability}) can be rewritten as
\begin{align}
\label{explicitP}
P(\vn,\vm) = \sum_{\s,\r \in S_N}\prod_{j=1}^N (V_{\s_j,j}V^*_{\r_j,j}\hS_{\r_j\s_j}).
\end{align} 
Applying the inclusion-exclusion principle to this equation, we obtain an algorithm to compute the probability \cite{tichy2015sampling} that is similar to Ryser's formula \cite{ryser1963combinatorial}:
\begin{align}
\label{ieryser}
P(\vn,\vm) = \sum_{\substack{S,R \subseteq \\ \{1,\dots ,N \} } } (-1)^{|S|+|R|} \prod_{j=1}^N \sum_{\substack{r\in R\\ s\in S} } V_{sj}V^*_{rj}\hS_{rs},
\end{align}
($|S|$  represents the number of elements for a given set $S$),
or equivalently,
\begin{align}
\label{pryser}
P(\vn,\vm)= \sum_{\vec{x},\vec{y} \in \{0,1\}^N}&(-1)^{\sum_i^N x_i + \sum_i^N y_i} \nn \\
&\times  \prod_{j=1}^N \Big[\sum_{r,s=1}^N V_{sj}x_sV_{rj}^*y_r \hS_{rs} \Big].
\end{align} 
The above identities directly result in the following feature for two extremal situations of $\hs$:
\begin{theorem}
	The transition probability $P(\vn,\vm)$ is the same for all maximally coherent $\hs$, i.e, $P(\vn,\vm)$ is equivalently hard to simulate for the cases.	If $\hs$ is incoherent, $P(\vn,\vm)$ is approximated efficiently.
\end{theorem}

\begin{proof}
	$\hs$ is maximally coherent if and only if $\hs = |\p\>\<\p|$ with $|\p\> = \frac{1}{\sqrt{N}}\sum_{i}e^{i\th_i}|i\>$ \cite{bai2015maximally}, or
	$\hS_{ij} = e^{i(\th_i-\th_{j})}$. Therefore, from the relation $ \prod_{i}\hS_{i\s_i} = \exp[i\sum_i(\th_i-\th_{\s_i})]=1$ for all $\s$, we can see that $P(\vn,\vm)$ is the same for all maximally coherent $\hS$. On the other hand, the only possible incoherent state of $\hs$ is when $\hs _{ij} = \d_{ij}/N$. Then $P(\vn,\vm)$ becomes a permanent of the nonnegative matrix (the unitary condition of $M$ moreover makes $M\odot M$ a doubly stochastic matrix), which is efficiently approximated \cite{jerrum2004polynomial}.
\end{proof}	

For arbirarily distinguishable photons, the classical runtime $\cT$ for the simulation with the algorithm Eq. \eqref{pryser} is given by $\cT =2^{2(N-1)}N^3$. However, this runtime can decrease when some elements of $\hS$ are zero, i.e., $\mathcal{N}(\hs) < N^2$. Indeed, the functional form of Eq. \eqref{pryser} shows that the runtime becomes 
\begin{align}
\cT = (2^{2(N-1)}N) \mathcal{N}(\hs)
\end{align} since the number of arithmetics in the bracket of Eq. \eqref{pryser} is $\mathcal{N}(\hs)$. This shows that the depletion of coherence decreases the computational cost of the exact simulation (see Figure 2 for $N=8$ example). 

\begin{figure}[htbp]
	\centering
	\includegraphics[width=7.5cm]{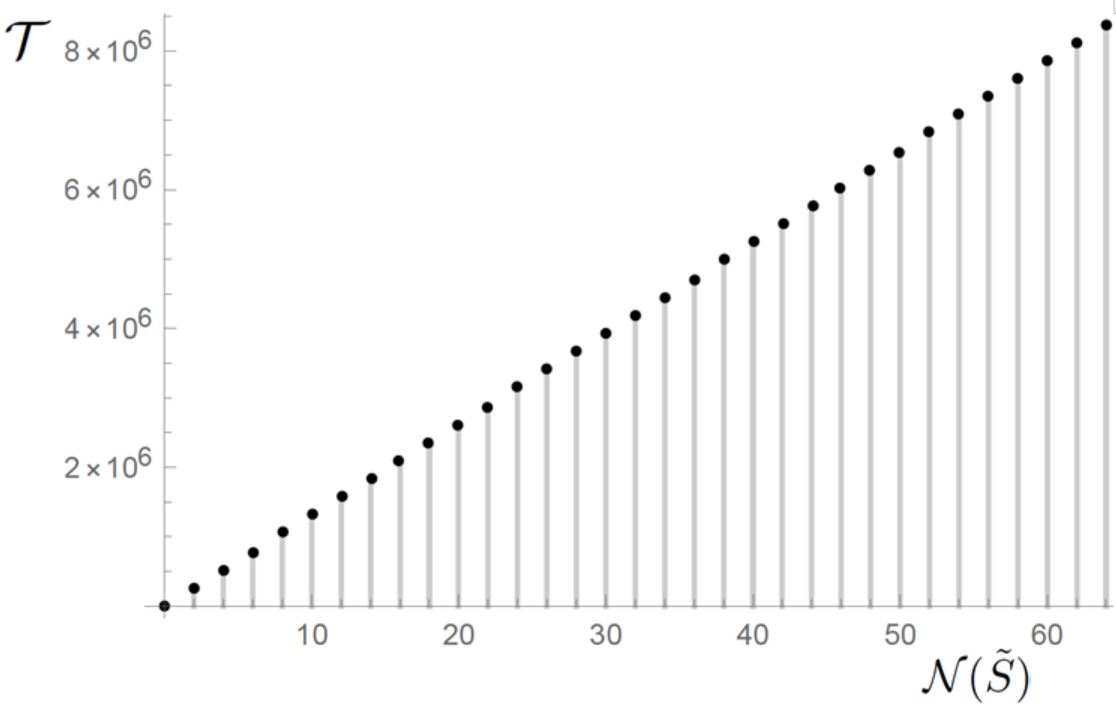}
	\caption{Runtime of exactly simulating $P(\vn,\vm)$ along $\mathcal{N}(\tilde{\hS})$ for $N=8$. $\cT$ has the maximal value $2^{14}8^3=8388608$ when $\mathcal{N}(\tilde{\hS}) = 8^2$. $\cT$ decreases monotonically as $\mathcal{N}(\tilde{S})$ decreases.}
\end{figure}

On the other hand, this pattern has a singular point when particles are completely indistinguishable, which correspond to the case when $\hs$ is maximally coherent. Then the transition probability becomes $|\perm(V)|^2$, the absolute square of the transition amplitude. For this case, the runtime is $2^{(N-1)}N^2$ from Ryser's formula. This abrupt decrease of the computational cost is due to the symmetry of $\hs$ for the maximally coherent case, which permits us find an algorithm with a shorter runtime.
To sum up, the classical runtime for the exact simulation is affected by two factors, the symmetry of $\hs$ (a classcal factor) and coherence (a quantum factor). \emph{When the effect of the symmetry disappears, the depletion of coherence decreases the computational time cost of the exact simulation.}

Another approach to reduce the number of arithmetic operations is to break up the subsets $S$ and $R$ defined in Eq. \eqref{ieryser} so that the corresponding submatrices of $\hS$ become zero. In other words, for some $S=\{s_1,\dots, s_\a \}$ and $R=\{r_1,\dots ,r_\b \}$, if $S_{s_i r_j}=0$ for all $i$ and $j$, we do not need to include the summation in the algorithm, which results in a shorter runtime, not significantly though. A specific example for $N=4$ is given in Appendix \ref{alt}.


\section{Conclusions}\label{con}

In this work, we showed that the partial distinguishability of photons in LON can be understood from the framework of coherence resource theory. We introduced the concept of \emph{permuted genuinely incoherent operation} (pGIO) that transforms one partial distinguishability matrix $\hS$ to another. By delineating the role of three pG coherence monotones ($\mathcal{N}(\hs)$, $\perm(|\hs|)$ and $J_\s$) in partially distinguishable boson sampling,  we presented evidence of the assumption that the coherence of partial distinguishability affects the computational time complexity of a partially distinguishable scattering process of linear optical network.

Our current work can develop in various directions.
For example, $\mathcal{N}$ decreases the runtime for exact simulation of transition probability with our current algorithm; however, the runtime is still an exponential function of $N$. There might exist more efficient algorithms that exploit the coherence of partial distinguishability to reduce runtime.
Also, the application of our analysis to the continuous BS system \cite{lund2014boson, rahimi2015, rahimi2016sufficient, huh2017vibronic, hamilton2017gaussian} would provide more rigorous conditions for various types of BS to be computationally hard.

\section*{Acknowledgements}

The authors are grateful to Jelmer Renema for his helpful comments. This work is supported by Basic Science Research Program through the National Research Foundation of Korea (NRF) funded by the Ministry of Education, Science and Technology (NRF-2015R1A6A3A04059773,  NRF-2019R1I1A1A01059964).

\appendix
\section{The upper bound of transition probability with $\perm(|\hS|)$}\label{bound}

A slight modification of Eq. (51) in Ref. \cite{tichy2015sampling} gives
\begin{align}
\label{probbound}
P(\vn,\vm) &= \Big|\sum_{\s\in S_N }\perm(V\odot V^*_{\s, \mathbb{I}})\Big(\prod_i \hS_{i, \s_i} \Big)\Big| \nn \\
&\le \perm(V\odot V^*)\sum_{\s \in S_N}\Big| \prod_i \hS_{i,\s_i}\Big|
\equiv  P_\mathbb{I}[\perm(|\hS|)].
\end{align} 
where the inequality comes from the relation $\big|\perm(V\odot V^*_{\mathbb{I},\s} )  \big| \le \perm(V\odot V^*)\equiv  P_{\mathbb{I}}$ for any permutation $\s$. 
Using the monotonicity of $\perm(|\hS|)$ from Theorem \ref{monotones}, we can see that a pGIO on $\hS$ decreases the upper bound of the transition probability $P(\vn,\vm)$. 
The unitarity condition of $V$ in Eq. \eqref{probbound} provides a more rigorous upper bound condition for the equation. Indeed, since $V\odot V^*$ is a unistochastic matrix (a doubly stochastic matrix whose elements are the absolute squares of the elements of a unitary matrix),  the upper and lower bounds for $P_{\mathbb{I}} = \perm(V\odot V^*)$ are given using the result in Ref. \cite{gurvits2014bounds} by 
\begin{align}\label{gurvitsbounds}
F(V\odot V^*) \le \perm(V\odot V^*) \le 2^NF(V\odot V^*) , 
\end{align}
where $F(V\odot V^*) \equiv \prod_{i,j=1}^N(1-|V_{ij}|^2)^{1-|V_{ij}|^2 }$. Hence, Eq. \eqref{probbound} can be rewritten as
\begin{align}
P(\vn,\vm) \le 2^N F(V\odot V^*)[\perm(|\hS|)]. 
\end{align}

\section{Alternative algorithm example}\label{alt}

($N=4$) The runtime using Eq. \eqref{ieryser} is $(2^{6}4^3)/2=2048$. However, when $S_{13}=S_{24}=S_{34}=0$, the summations with the following $(R,S)$ become zero:
\begin{align}
(R,S)=&(\{1\}, \{3 \}),\quad (\{2\}, \{4 \}), \nn \\
&(\{3\}, \{4 \}),\quad (\{3\}, \{1,4 \}), \quad (\{4\}, \{2,3 \}),
\end{align}
which contains 4, 4, 8, and 8 terms, respectively. Therefore, the resulting runtime decreases to 2024.

\bibliographystyle{unsrt}
\bibliography{pdbs2}

\end{document}